\documentclass[12pt]{article}
\usepackage{fullpage}
\usepackage[top=2cm, bottom=4.5cm, left=2.5cm, right=2.5cm]{geometry}
\usepackage{amsmath,amsthm,amsfonts,amssymb,amscd}
\usepackage{lastpage}
\usepackage{enumerate}
\usepackage{fancyhdr}
\usepackage{mathrsfs}
\usepackage{xcolor}
\usepackage{xparse}
\usepackage{graphicx}
\usepackage{listings}
\usepackage{amssymb}
\usepackage{relsize}
\usepackage{centernot}
\usepackage{tasks}
\usepackage{mathtools}
\usepackage{braket}
\usepackage[backend=biber,style=numeric, sorting=none]{biblatex}
\bibliography{refs}

\usepackage[breaklinks]{hyperref}
\usepackage{breakurl}

\hypersetup{%
  colorlinks=true,
  linkcolor=blue,
  linkbordercolor={0 0 1}
}
 
\title{Trotterization in Quantum Theory}
\author{Physics Claire Kluber}
\date{}

\newcommand\paren[1]{ \left( #1 \right) }
\newcommand\abs[1]{\left \lvert #1 \right \rvert}
\newcommand{\norm}[1]{\left\lVert  #1\right\rVert}

\NewDocumentCommand{\evalat}{sO{\big}mm}{%
  \IfBooleanTF{#1}
   {\mleft. #3 \mright|_{#4}}
   {#3#2|_{#4}}%
}

\DeclareMathOperator\slim{s-lim}

\newtheorem*{claim}{Claim}

\begin{document}

\maketitle

\section*{Abstract}
Trotterization in quantum mechanics is an important theoretical concept in handling the exponential of noncommutative operators. In this communication, we give a mathematical formulation of the Trotter Product Formula, and apply it to basic examples in which the utility of Trotterization is evident. Originally, this article was completed in December 2020 as a report under the mentorship of Esteban Cárdenas for the University of Texas at Austin Mathematics Directed Reading Program (DRP). However, the relevance of Trotterization in reducing quantum circuit complexity has warranted the release of a revised and more formal version of the original. Thus, we present a mathematical perspective on Trotterization, including a detailed sketch of a formal proof of the Trotter Product Formula. 

\tableofcontents{}

\section{Introduction}
One of the major goals of quantum mechanics is finding solutions, called wavefunctions/eigenfunctions, to the time-independent Schr\"odinger wave equation. For a given time-independent Hamiltonian operator $\hat{H}$ on a Hilbert space $\mathcal{H}$, the Schr\"odinger equation is given by
\begin{equation} \label{eq:schro}
    \hat{H}\ket{\Psi} = E\ket{\Psi}
\end{equation}
where $\ket{\Psi} \in \mathcal{H}$ is an eigenfunction and $E \in \mathbb{R}$ is an energy eigenvalue. Alone, this equation only yields the energy values and the stationary states of the physical system. To get the time-dependent eigenvector $\ket{\Psi (t)}$, one needs the unitary operator $U(t)$:
\begin{gather}
    U(t) = e^{-i\hat{H}t/\hbar}, \label{eq:unitary} \\
    \ket{\Psi (t)} = U(t) \ket{\Psi} = e^{-i\hat{H}t/\hbar} \ket{\Psi}, \label{eq:time-dep}
\end{gather}
where $\hbar = h / 2\pi$ is Planck's reduced constant.

In practice, it can be hard to compute this operator exponential, so let us focus on the simplest example: When the (finite-dimensional in this case) Hamiltonian operator $\hat{H}$ is given by the diagonal matrix
\begin{equation} \label{eq:fin-dim}
    \hat{H} = \begin{pmatrix}
        E_1 & 0 & \cdots & 0 \\
        0 & E_2 & \cdots & 0 \\
        \vdots & \vdots & \ddots & \vdots \\
        0 & 0 & \cdots & E_n
    \end{pmatrix}.
\end{equation}
Substituting this into Equation \ref{eq:schro}, the eigenfunctions are given by the standard basis
\begin{equation}
    \ket{\psi_1} = \begin{pmatrix} 1 \\ 0 \\ \vdots \\ 0 \end{pmatrix}, 
    \quad \ket{\psi_2} = \begin{pmatrix} 0 \\ 1 \\ \vdots \\ 0 \end{pmatrix}, \quad \cdots,
    \quad \ket{\psi_n} = \begin{pmatrix} 0 \\ 0 \\ \vdots \\ 1 \end{pmatrix}
\end{equation}
with corresponding eigenvalues $E_1, \cdots, E_n$, respectively. From this, an arbitrary stationary state $\ket{\Psi}$ of the entire quantum system is simply a complex ($c_i \in \mathbb{C}$) linear combination:
\begin{equation}
    \ket{\Psi} = \sum_{i=1}^n c_i \ket{\psi_i}, \quad \sum_{i=1}^n \abs{c_i}^2 = 1.
\end{equation}
This is the time-independent solution. To get to the full time-dependent solution $\ket{\Psi(t)}$, we simply need to compute $U(t)$. This is made easy by the fact that the exponential of a diagonal matrix is just the diagonal matrix of element exponentials. Concretely,
\begin{align}
    U(t) &= \exp\left[\begin{pmatrix}
        -iE_1t/\hbar & 0 & \cdots & 0 \\
        0 & -iE_2t/\hbar & \cdots & 0 \\
        \vdots & \vdots & \ddots & \vdots \\
        0 & 0 & \cdots & -iE_nt/\hbar
    \end{pmatrix}\right] \nonumber \\
    &= \begin{pmatrix}
        e^{-iE_1t/\hbar} & 0 & \cdots & 0 \\
        0 & e^{-iE_2t/\hbar} & \cdots & 0 \\
        \vdots & \vdots & \ddots & \vdots \\
        0 & 0 & \cdots & e^{-iE_nt/\hbar}
    \end{pmatrix}.
\end{align}
Therefore, the time-dependent state $\ket{\Psi(t)}$ is given by
\begin{align}
    \ket{\Psi(t)} &= U(t) \ket{\Psi} \nonumber \\
    &= \begin{pmatrix}
        e^{-iE_1t/\hbar} & 0 & \cdots & 0 \\
        0 & e^{-iE_2t/\hbar} & \cdots & 0 \\
        \vdots & \vdots & \ddots & \vdots \\
        0 & 0 & \cdots & e^{-iE_nt/\hbar}
    \end{pmatrix} \paren{\sum_{i=1}^n c_i \ket{\psi_i}} \nonumber \\
    &= \sum_{i=1}^n c_i e^{-iE_it/\hbar} \ket{\psi_i}.
\end{align}

This calculation was simple because the Hamiltonian $\hat{H}$ was diagonal, but for the general case we need to use the series definition of exponentiation. Given a (potentially infinite-dimensional) self-adjoint (Hermitian) operator $S$ defined on $\mathcal{H}$, its exponential is defined by
\begin{equation} 
    e^S = \sum_{n=0}^\infty \frac{S^n}{n!}.
\end{equation}
When $S$ is finite-dimensional, simply diagonalizing the matrix is sufficient to compute $e^S$, and therefore solve the eigenvalue problem. However, when $S$ is infinite dimensional, the problem becomes much more complicated. To help reduce this complexity, we can consider decomposing an operator into a sum. 

If we have another self-adjoint operator $T$ defined on $\mathcal{H}$ such that $[S,T] = 0$ (i.e., $ST = TS$), then the exponential of the sum splits:
\begin{equation}
    e^{S+T} \xi = e^S e^T \xi, \quad \forall\xi \in \mathcal{H} \label{eq:split}
\end{equation}
We can actually prove this in a relatively straight-forward manner from the definition of the operator exponential by applying the Binomial Theorem:
\begin{align*}
    e^{S+T} &= \sum_{n=0}^\infty \frac{(S+T)^n}{n!} \\
    &= \sum_{n=0}^\infty \sum_{k=0}^n \frac{1}{n!} \binom{n}{k} S^{n-k}T^k \\
    &= \sum_{n=0}^\infty \sum_{k=0}^n \frac{1}{n!} \frac{n!}{(n-k)!k!} S^{n-k}T^k \\
    &= \sum_{n=0}^\infty \sum_{k=0}^n \frac{1}{(n-k)!k!} S^{n-k}T^k.
\end{align*}
With this result, notice that every possible product of $S^m$ with $T^n$ occurs for $m,n \in \mathbb{Z}^+ \cap \{0\}$. Thus, rewrite the sum as follows:
\begin{align*}
    \sum_{n=0}^\infty \sum_{k=0}^n \frac{1}{(n-k)!k!} S^{n-k}T^k &= \sum_{n=0}^\infty \sum_{m=0}^\infty \frac{1}{m!n!} S^mT^n \\
    &= \paren{\sum_{m=0}^\infty  \frac{1}{m!} S^m} \paren{\sum_{n=0}^\infty  \frac{1}{n!} T^n} \\
    &= e^S e^T.
\end{align*}

In the case where $S$ and $T$ do not commute, this argument fails because the Binomial Theorem no longer applies. It seems like there's no way to generalize this argument for the noncommutative case. For example, in the binomial expansion of $(S+T)^3$ with $[S,T] \ne 0$, $STS \ne S^2 T$, and so it's impossible to collect terms on the left and right sides of the overall sum. A somewhat surprising result called the Trotter Product Formula is needed for noncommutative operators, which in a strong sense approximates the above splitting of the exponential (Equation \ref{eq:split}). This is what we explore for the remainder of this communication.

Before the next section, it's important to bring up one prerequisite: unitary evolution groups. Put simply, a unitary evolution group is a group of unitary operators $G(t)$ for $t \in \mathbb{R}$ given by a homomorphism of $(\mathbb{R}, +)$ (i.e., for all $s,t \in \mathbb{R}$, $G(s+t) = G(s)G(t)$). An important notion is the infinitesimal generator $T$ of a unitary evolution group, given pointwise by
\begin{equation} \label{eq:inf-gen} 
    T \xi = i \lim_{h \xrightarrow{} 0} \frac{G(h) - I}{h} \xi.
\end{equation}
It turns out that (with some technical regularity of the mapping $t \to G(t)$, namely weak measurability) $T$ is self-adjoint (by Stone's Theorem on one-parameter unitary groups). Furthermore, it also turns out that any self-adjoint operator $T$ corresponds to the (strongly/pointwise continuous) unitary evolution group $U_T(t) = e^{-itT}$. This correspondence will be useful in formally proving the Trotter Product Formula.

Furthermore, note that the domain for $S$ and $T$ was always taken to be the entire Hilbert space. This is because if an operator $S$ is bounded and self-adjoint, then its dense domain $D(S)$ can be extended to all of $\mathcal{H}$ by constructing a new unique extended operator $\overline{S}$ from $S$ for which $\overline{S} = S$ on $D(S)$. In general, domains only matter for unbounded operators, which we consider in the next section in the statement of the Trotter Product Formula proof. 

\section{The Trotter Product Formula}
The following proof is adapted from de Oliveria's textbook presentation \cite{oliveira}. 
\begin{claim}
    Let $S$ and $T$ be (potentially unbounded) self-adjoint operators on $\mathcal{H}$. Then, for every $t \in \mathbb{R}$ and for all $\xi \in D(S+T) = D(S) \cap D(T) = \mathcal{D}$,
    \begin{equation} \label{eq:claim}
        \lim_{n\xrightarrow{}\infty} \norm{e^{-it(S+T)}\xi - \paren{e^{-i\frac{t}{n}S}e^{-i\frac{t}{n}T}}^n\xi} = 0.
    \end{equation}
    In other words, the following strong (pointwise) operator limit holds:
    \begin{equation*}
        \smashoperator\slim\limits_{n\to\infty} \paren{e^{-i\frac{t}{n}S}e^{-i\frac{t}{n}T}}^n = e^{-it(S+T)}.
    \end{equation*}
\end{claim}
\begin{proof}    
    Let $h \in \mathbb{R}$ such that $h \ne 0$ and $\xi \in \mathcal{D}$; then, define $u_h(\xi)$ as 
    \begin{equation*}
        u_h(\xi) = \frac{1}{h} \paren{e^{-ihS}e^{-ihT} - e^{-it(S+T)}}.
    \end{equation*}
    Now, this can be rewritten as
    \begin{equation*}
        u_h(\xi) = \frac{e^{-ihS} - I}{h} \xi + e^{-ihS}\frac{e^{-ihT} - I}{h} \xi - \frac{e^{-ih(S+T)} - I}{h} \xi
    \end{equation*}
    where $I$ is the identity operator on $\mathcal{H}$. Now, we can use the fact that these resemble unitary evolution group generators (Equation \ref{eq:inf-gen}). Taking the limit $h \xrightarrow{} 0$ for the first and third terms yields $S \xi$ and $-(S+T)\xi$, respectively. For the second term, the Dominated Convergence Theorem \cite{stein-real} applies, so that it yields $T \xi$ as $h \xrightarrow{} 0$. Thus, $\lim_{h\xrightarrow{} 0} u_h(\xi) = 0$.
    
    Since $u_h$ is linear and bounded, and since the operator $S+T$ is closed (because the sum of two self-adjoint operators is self-adjoint and every self-adjoint operator is closed), one can show through the Uniform Boundedness Principle, as in de Oliveria's work \cite{oliveira}, that
    \begin{equation}
        \lim_{h \xrightarrow{} 0} \sup_{\abs{s} < \abs{t}} \norm{u_h(\xi_s)} = 0,
    \end{equation}
    where $\xi_s$ is given by the unitary evolution group $\xi_s = e^{-is(S+T)}$. Because the unitary evolution group is strongly continuous, it follows that $J_{\xi, t} = \set{\xi_s | \abs{s} \leq \abs{t}}$ is compact and totally bounded in $J_{\xi, t}$ under the graph norm of $S+T$. Because $J_{\xi, t}$ is totally bounded, it can be covered by a finite number of open balls. This leads to an interpolation argument: any $\xi_s \in J_{\xi, t}$ lies inside one of those balls, and from that we can upper bound the graph norm with a distance term that vanishes as $h \xrightarrow{} 0$. Thus, since $\xi_s \in J_{\xi, t}$ is arbitrary, this bound also applies to the $\xi_s$ that maximizes $\norm{u_h(\xi_s)}$. 
    
    Now that we have shown this, we need to relate $u_h(\xi)$ with the difference in Equation \ref{eq:claim}. For bounded operators $A, B$ and any $n \in \mathbb{Z}^+$, one can show that
    \begin{equation*}
        A^n - B^n = \sum_{j=0}^{n-1} A^j (A-B)B^{n-1-j}. 
    \end{equation*}
    Substituting $A = e^{-it(S+T)/n}$ (so that $A^n = e^{-it(S+T)}$) and similarly with $B = e^{-i\frac{t}{n}S}e^{-i\frac{t}{n}T}$, 
    \begin{align*}
        &\paren{e^{-it(S+T)/n}}^n - \paren{e^{-itS/n}e^{-itT/n}}^n \\
        &= \sum_{j=0}^{n-1} \paren{e^{-itS/n}e^{-itT/n}}^j \paren{e^{-it(S+T)/n}-e^{-itS/n}e^{-itT/n}}\paren{e^{-it(S+T)/n}}^{n-1-j}
    \end{align*}
    Therefore, we can take the norm and apply it to an arbitrary $\xi \in \mathcal{D}$. Applying the triangle inequality along with the fact that the norm of unitary operators is unity, 
    \begin{align*}
        &\norm{\sum_{j=0}^{n-1} \paren{e^{-itS/n}e^{-itT/n}}^j \paren{e^{-it(S+T)/n}-e^{-itS/n}e^{-itT/n}}\paren{e^{-it(S+T)/n}}^{n-1-j}\xi} \\
        &\leq \sum_{j=0}^{n-1} \norm{\paren{e^{-itS/n}e^{-itT/n}}^j \paren{e^{-it(S+T)/n}-e^{-itS/n}e^{-itT/n}}\paren{e^{-it(S+T)/n}}^{n-1-j}\xi} \\
        &\leq \sum_{j=0}^{n-1} \norm{\paren{e^{-itS/n}e^{-itT/n}}^j}\norm{\paren{e^{-it(S+T)/n}-e^{-itS/n}e^{-itT/n}}\paren{e^{-it(S+T)/n}}^{n-1-j}\xi} \\
        &= \sum_{j=0}^{n-1} \norm{ \paren{e^{-it(S+T)/n}-e^{-itS/n}e^{-itT/n}}\paren{e^{-it(S+T)/n}}^{n-1-j}\xi} 
    \end{align*}
    Now, the only dependence we have on $j$ is in the last term of the product. If we define $s_j = t(n-1-j)/n$, then that term becomes $e^{-is_j(S+T)}$. Because $\abs{s_j} < \abs{t}$, $\{s_j\}$ is a subset of the interval $[-\abs{t}, \abs{t}]$. Therefore, we upper bound the previous expression with a supremum over the entire interval:
    \begin{align*}
        &\sum_{j=0}^{n-1} \norm{ \paren{e^{-it(S+T)/n}-e^{-itS/n}e^{-itT/n}}e^{-is_j(S+T)}\xi} \\ 
        &\leq  n \sup_{\abs{s} < \abs{t}} \norm{ \paren{e^{-it(S+T)/n}-e^{-itS/n}e^{-itT/n}}e^{-is(S+T)}\xi} \\
        &= n \sup_{\abs{s} < \abs{t}} \norm{ \paren{e^{-it(S+T)/n}-e^{-itS/n}e^{-itT/n}}\xi_s} 
    \end{align*}
    If we let $h = \abs{t}/n$, then the last expression becomes
    \begin{align*}
        &\frac{\abs{t}}{h}\sup_{\abs{s} < \abs{t}} \norm{ \paren{e^{-it(S+T)/n}-e^{-itS/n}e^{-itT/n}}\xi_s} \\
        &=\abs{t}\sup_{\abs{s} < \abs{t}} \norm{ \frac{1}{h}\paren{e^{-it(S+T)/n}-e^{-itS/n}e^{-itT/n}}\xi_s} \\
        &= \abs{t}\sup_{\abs{s} < \abs{t}} \norm{u_h(\xi_s)}.
    \end{align*}
    Therefore, as $n \xrightarrow{} \infty$, $h \xrightarrow{} 0$ so that 
    \begin{equation*}
         \paren{e^{-i\frac{t}{n}S}e^{-i\frac{t}{n}T}}^n\xi \xrightarrow{} e^{-it(S+T)}\xi, \quad \xi \in \mathcal{D}.
    \end{equation*}
\end{proof}

\section{Applications and Implications}
Given a finite-dimensional $n \times n$ (complex) Hermitian matrix $A$, its matrix exponential can be computed exactly by diagonalizing. By the complex spectral theorem for finite-dimensional matrices, $A$ can be decomposed as $A = UDU^\dag$, where $U$ is a unitary matrix, $D$ is a diagonal matrix, and $\paren{\cdot}^\dag$ denotes the conjugate transpose. Using the fact that $A^m = (UDU^\dag)^m = UD^m U^\dag$ for any $m \in \mathbb{Z}^+$, it turns out that $e^A$ is given by
\begin{equation} \label{eq:slow}
    e^A = Ue^D U^\dag.
\end{equation}
However, when the dimensionality $n$ is very large, diagonalizing $A$ can be computationally difficult. Therefore, an approximation is needed to compute matrix exponentials in general instead of using Equation \ref{eq:slow} directly. 

As a very simple example, we can approximate the matrix exponential $e^{A+B}$ by Trotterization for noncommuting matrices $A$ and $B$. Using the series definition of the matrix exponential for sufficiently large $N$,
\begin{equation*}
    e^{A/N} \approx I + \frac{A}{N}, \quad e^{B/N} \approx I + \frac{B}{N}.
\end{equation*}
As a consequence of the Trotter Product Formula,
\begin{align*}
    e^{A+B} &\approx \paren{e^{A/N}e^{B/N}}^N 
    &\approx \left[\paren{I + \frac{A}{N}}\paren{I + \frac{B}{N}}\right]^N 
    &=\left[I + \frac{A}{N} + \frac{B}{N} + \frac{AB}{N^2}\right]^N 
    &\approx \left[I + \frac{A+B}{N}\right]^N.
\end{align*}
Matrix multiplication is still $O(n^3)$, but it is generally faster and more parallelizable than diagonalization. Furthermore, this can be implemented for large $N = 2^m$ using a repeated squaring algorithm \cite{repeated-squaring}. Another interesting observation is that the approximation is commutative with respect to the inputs $A$ and $B$. As a final thought, note the similarities between this equation and the limit definition for $e^x$ for $x \in \mathbb{R}$:
\begin{equation*}
    e^x = \lim_{n \xrightarrow{} \infty} \paren{1+\frac{x}{n}}^n.
\end{equation*}
For a more practical example of this principle applied, specifically a Trotterized Hamiltonian equation, see Whitfield et al. (2012) \cite{whitfield}.

Trotterization is also useful for analyzing operators that individually have "nice" operator exponentials, but when summed together have a complicated operator exponential. A good example of this is the unitary evolution group $e^{it(\hat{p}^2/2m + \hat{x})}$ corresponding to the Hamiltonian operator $\hat{H} = \frac{\hat{p}^2}{2m} + \hat{x}$ (i.e., $V(x) = x$). It is difficult to compute this exponential directly, but the exponential of $it\hat{p}^2/2m$ individually corresponds to the free Schr\"odinger kernel (which is well-understood mathematically), and the exponential of $it\hat{x}$ individually corresponds to a momentum translation. Thus, the exponential of the sum can be reduced to terms much easier to analyze.

Note that by no means is Trotterization the optimal decomposition of these operators. Trotterization is to an optimal decomposition as a Taylor expansion is to a least-squares curve. Therefore, it can be shown that there are more optimal circuits with as many or fewer gates \cite{lubasch}. 

\section{Conclusion}
Thus, we have proved the Trotter Product Formula in its full generality. Additionally, we have discussed its motivations from a practical perspective within quantum mechanics. This combination is powerful because it means that Trotterization can both be applied to practical quantum mechanics challenges, such as computing quantum circuits, and theoretical challenges, such as reducing quantum circuit complexity.

In some cases, the Trotter Product Formula admits a simpler form. When operators $A,B$ (on some Hilbert space $\mathcal{H}$) are anti-commutative, such that $AB = -BA$, then $e^{A+B}$ can be calculated using a generalization of the Binomial Theorem to q-commutative ($q=-1$) algebras. A detailed reference and derivation for this can be found in Scurlock (2020) \cite{Scurlock}. This can be practically applied to, for example, Pauli gate operations. A detailed work is utilizing this result is Zhao \& Yuan (2021) \cite{applied-anticommutative}.

\small
\printbibliography

\end{document}